\documentclass[journal]{IEEEtran}
\usepackage{booktabs}

\ifCLASSINFOpdf
 \usepackage[pdftex]{graphicx}
 \usepackage{epstopdf}
\DeclareGraphicsExtensions{.pdf}
\else
\usepackage[dvips]{graphicx}
\DeclareGraphicsExtensions{.eps}
\fi
\usepackage[cmex10]{amsmath}
\usepackage{amsthm,amssymb}
\usepackage{tabularx}
\newtheorem{proposition}{Proposition}

\usepackage{algorithm,algorithmic,multicol}
\begin{document}
\title{Adaptive Windowing for ICI Mitigation in Doubly Selective Channels with Unknown Statistics}

\author{Evangelos~Vlachos
        and~Kostas~Berberidis,~\IEEEmembership{Senior Member,~IEEE}% <-this % stops a space
%\thanks{...}% <-this % stops a space
%\thanks{E.Vlachos and K.Berberidis are with University of Patras.}% <-this % stops a space
}

\maketitle

\begin{abstract}
In doubly selective channels, receiver windowing constitutes an effective technique for enhancing the banded structure of the frequency-domain channel matrix, and thus improving the effectiveness of a banded equalizer for intercarrier interference (ICI) mitigation. A common window design technique, which performs close to optimal, is based on the criterion of maximum average signal-to-interference-plus-noise ratio (SINR). The optimality of this technique has been verified for stationary channels with perfectly known statistics. However, in cases where this assumption does not hold, a near optimal performance can be achieved at the expense of high complexity cost. To overcome these limitations, an adaptive windowing technique is proposed that is able to track the optimal receiver window  offering low-complexity requirements. Through simulation experiments it has been verified that the proposed technique is able to adapt to the varying channel statistics with increased robustness to channel modeling errors.

\end{abstract}

%\begin{IEEEkeywords}
%\end{IEEEkeywords}


\section{Introduction}
 
\IEEEPARstart{I}{n} orthogonal frequency division multiplexing (OFDM) systems with high levels of mobility, the experienced channels are usually both time- and frequency-selective (so-called \textit{doubly selective}) \cite{Stantchev00}. The temporal variations within one OFDM block corrupt the subchannels orthogonality, generating power leakage among the subcarriers, thus causing intercarrier interference (ICI) at the receiver. To mitigate the ICI effect, a block banded minimum-mean-square-error (MMSE) equalizer \cite{Hlawatsch11} can be employed, which exploits the special structure of the frequency-domain channel convolution matrix, via band matrix approximation \cite{Rugini06}, \cite{Jeon99}, \cite{Vlachos13}. However, this approximation results into severe performance loss, especially in regimes with high signal-to-noise ratio (SNR), \cite{Rugini07}.

An effective technique, which is able to enhance the performance of banded equalizers, is to perform a time-domain pre-filtering of the input signal (\textit{windowing}), in order to  enforce a banded structure of the matrix \cite{DSilva02},\cite{Schniter04},\cite{Rugini05}. Let us distinguish two categories for window design; the first one is based on a predefined window (e.g. Hamming) while the second one is based on the maximization of a signal-to-interference-plus-noise ratio (SINR) criterion (e.g. maximum SINR or maximum average SINR,  \cite{Schniter04}). Among these methods, the \textit{maximum SINR} exhibits the best performance, but it requires perfect knowledge of the channel impulse response (CIR) and high computational burden. On the other hand, the predefined windowing exhibits poor performance, especially at high SNR. In this letter, we focus on criteria based on the maximization of the SINR, and in particular, on the maximum average-SINR. This criterion been employed for stationary channels, with perfectly known channel statistics, resulting into almost optimal performance \cite{Schniter04}. However, when these conditions do not hold, maximum average-SINR technique cannot be employed in a straightforward manner. To overcome this limitation, we propose to estimate the channel statistics in an adaptive manner, on an OFDM block basis. However, the direct solution of the resulting maximum SINR problem would require cubic computational complexity order over the number of the subcarriers. On this premise, we propose a novel \textit{adaptive windowing} technique which is able to track the variations of the channel statistics offering at the same time low-complexity requirements. Moreover, the proposed technique exhibits enhanced robustness over the channel modeling errors.

\textit{Notation:} $\mathcal{E}\{\cdot\}, \mathcal{D}(\cdot), \mathcal{C}(\cdot)$ denote the statistical mean, the diagonal matrix and the circulant matrix of the argument respectively; $\mathcal{N}(0, \sigma^2)$ denotes the zero-mean additive complex white Gaussian distribution; $[X]_{i,j}$ denotes $(i,j)$-th element of the matrix $\mathbf{X}$; $\circ$ denotes the Hadamard (element-wise) product, $\mathbf{T}_K$ is a matrix with lower and upper bandwidth $K/2$ and all ones within its band; $\mathbf{F}$ denotes the discrete Fourier transform (DFT) matrix; $\mathbf{I}_N$ denotes the $N \times N$ identity matrix.

\section{Problem Statement}

In an OFDM transmitter, the frequency-domain data stream is divided into blocks of length $N$ and modulated by $N$-point inverse DFT. At the receiver, the received blocks are demodulated by $N$-point DFT. Assuming time and frequency synchronization, and employing a cyclic prefix length greater than the maximum delay spread of the channel, the input-output relation for each OFDM block can be described as
\begin{equation}\label{eq:io}
\mathbf{y}= \mathbf{F}\mathbf{H}_t \mathbf{F}^H = \mathbf{H} \mathbf{x} + \mathbf{z}
\end{equation}
where $\mathbf{x}$ and $\mathbf{y}$ are the $N \times 1$ transmitted and received symbol vectors of each OFDM block, respectively, $\mathcal{E}\{ \mathbf{x}\mathbf{x}^H \} = \mathbf{I}_N$, $\mathbf{H}_t$ and $\mathbf{H}$ denote the channel convolution $N \times N$ matrices at the time and the frequency domain respectively, and $\mathbf{z}$ denotes the $N \times 1$ additive complex white Gaussian noise (AWGN) vector with $\mathbf{z} \sim \mathcal{N}(0, \sigma_z^2 \mathbf{I}_N), i=1,\ldots,N$. In time-selective channels, $\mathbf{H}$ is typically a non-diagonal matrix whose off-diagonal elements are due to ICI. Let us denote the $K$-banded approximation of the frequency-domain channel matrix as $\mathbf{H}_{B}=\mathbf{T}_K \circ \mathbf{H}$, with $K$ non-zero elements at each row. Therefore, the banded MMSE-based soft-decision symbol vector is given by
\begin{equation}\label{eq:MMSE_equalizer}
\mathbf{\tilde{x}} = \mathbf{H}_B^H  \left( \mathbf{H}_B \mathbf{H}_B^H + \rho \mathbf{I}_N \right)^{-1} \mathbf{y}
\end{equation}
where $\rho = \frac{\sigma_z^2}{\epsilon}$ and $\epsilon$ is a regularization parameter in order to improve the condition number of the equalizer matrix. Applying a time-domain window $\mathbf{w}$ at the receiver, prior to the DFT operation of each OFDM block, we get the following output
\begin{equation}\label{eq:io_win}
\mathbf{y}_{\mathbf{w}} = \mathcal{C}(\mathbf{w}) \mathbf{y} = \mathbf{F}\mathcal{D}(\mathbf{w})\mathbf{y}.
\end{equation}
In this case, the MMSE-based soft-decision output is given by
\begin{equation}\label{eq:MMSE_equalizer_with_window}
\mathbf{\tilde{x}}_{\mathbf{w}} = \mathbf{H}_B^H  \left( \mathbf{H}_B \mathbf{H}_B^H + \rho \mathcal{C}(\mathbf{w})\mathcal{C}(\mathbf{w})^H \right)^{-1} \mathbf{y}_{\mathbf{w}}.
\end{equation}
The two optimal design criteria, i.e. the maximum SINR (Max-SINR) and the maximum average SINR (Max-average-SINR), correspond to the solution of the following two optimization problems, respectively \cite{DSilva02}
\begin{eqnarray}\label{eq:maxSINR}
\mathbf{w}^{\star} &=& arg \max_{\mathbf{w}} \frac{P_s}{P_{ni}} \\ 
\mathbf{\bar{w}}^{\star} &=& arg \max_{\mathbf{w}} \frac{\mathcal{E}\{ P_s \}}{\mathcal{E}\{ P_{ni} \}} \label{eq:maxavSINR}
\end{eqnarray}
where $P_s = \Vert \mathbf{T} \circ \left( \mathcal{C}(\mathbf{w}) \mathbf{H} \right) \Vert_{F}^2$ is the signal power and $P_{ni} = \Vert \mathbf{T}^c\circ \left( \mathcal{C}(\mathbf{w}) \mathbf{H}\right) \Vert_{F}^2 + \sigma_z^2 \Vert \mathcal{C}(\mathbf{w}) \Vert_{F}^2$ is the noise plus interference power.

\section{Proposed Windowing Technique}

As mentioned in the Introduction, a common assumption, that simplifies the whole problem, is that the channel remains stationary. In this case, the channel statistics can be obtained once for all the OFDM block transmissions. When there are no estimation errors, the maximum average SINR criterion \eqref{eq:maxavSINR} can perform identically to \eqref{eq:maxSINR},  \cite{Schniter04}. 

However in practice, there are several cases where the channels are quasi- or even non-stationary \cite{Molisch09}, leading to an erroneously estimated channel. A straightforward solution to this problem, would be first to estimate the unknown channel statistics in an adaptive manner, and then solve the problem \eqref{eq:maxavSINR}. Although this would result into a more robust technique, it would also had a prohibited complexity (i.e. $\mathcal{O}(N^3)$) if the update is performed at every OFDM block. To overcome these difficulties, we propose an \textit{adaptive technique} which exhibits an order of magnitude lower complexity than the straightforward solution, and, at the same time, it is able to track the optimal average window. 

First we derive an alternative formulation for the optimal window design \eqref{eq:maxSINR}, which will stand as a basis for the subsequent analysis, as well as ground-truth models for the performance evaluation in Section~\ref{sec:Evaluation}.

\begin{proposition}[Max-SINR criterion]\label{proposition:1} The problem in \eqref{eq:maxSINR} can be equivalently expressed as a generalized eigenvalue problem, as follows,
	\begin{equation}\label{eq:maxSINR2}
	\mathbf{w}^{\star} = arg \max_{\mathbf{w}} \frac{\mathbf{w}^H \mathbf{R} \mathbf{w} }{ \mathbf{w}^H \left( \mathbf{\Lambda} - \mathbf{R} \right) \mathbf{w} }
	\end{equation}
	where $\mathbf{R} = \sum_{n=1}^N \left\{ \mathcal{D}(\mathbf{F}_n) \mathbf{F}  \mathbf{H} \mathcal{D}(\mathbf{T}_n) \mathbf{H}^H  \mathbf{F} ^H  \mathcal{D}(\mathbf{F}_n) \right\} $ and $\mathbf{\Lambda}=\mathcal{D}(\mathbf{F}\mathbf{H}\mathbf{H}^H \mathbf{F}^H) + \sigma_z^2 \mathbf{I}_N$ is a diagonal matrix.
\end{proposition}
\begin{proof}
	c.f. Appendix A.
\end{proof}

To proceed further, let us introduce a generalization of maximum average SINR criterion \eqref{eq:maxavSINR}.
\begin{proposition}[Generalized Max-average-SINR]\label{proposition:2}
The optimal window design based on the maximum average SINR criterion with unknown channel statistics, can be obtained by first solving the following maximum eigenvalue problem,
\begin{equation}\label{eq:maxSINR4}
\mathbf{\bar{v}}^{\star} = arg \max_{\mathbf{v}} \frac{\mathbf{v}^H \left( \mathcal{E} \{ \mathbf{\Lambda}\}^{-1/2} \mathcal{E} \{\mathbf{R} \} \mathcal{E} \{\mathbf{\Lambda}\}^{-1/2} \right) \mathbf{v}}{\mathbf{v}^H \mathbf{v}}
\end{equation}
and then by substituting to the expression
\begin{equation}
  \mathbf{\bar{w}}^{\star} = \mathcal{E} \{ \mathbf{\Lambda}\}^{-1/2} \mathbf{\bar{v}}^{\star}
\end{equation}
\end{proposition}
\begin{proof}
c.f. Appendix B.
\end{proof}
Considering that the CIR statistics is unknown, we proceed by approximating the correlation matrices $\mathcal{E} \{ \mathbf{\Lambda}\}$ and $\mathcal{E} \{ \mathbf{R}\}$ with the following sample-based expressions,
\begin{eqnarray}
\mathcal{E} \{ \mathbf{R} \} &\approx& \mathbf{\bar{R}}(m) \triangleq \frac{1}{m} \sum_{k=1}^m \lambda^{m-k} \mathbf{R}(m) \\
\mathcal{E} \{ \mathbf{\Lambda} \} &\approx& \mathbf{\bar{\Lambda}}(m) \triangleq  \frac{1}{m} \sum_{k=1}^m \lambda^{m-k} \mathbf{\Lambda}(m)
\end{eqnarray}
where $m$ is the OFDM block index. The sample sequences are defined as
\begin{equation}\label{eq:Rm}
\mathbf{R}(m) = \sum_{n=1}^N \left\{ \mathcal{D}(\mathbf{F}_n) \mathbf{F}  \mathbf{\tilde{H}}(m) \mathcal{D}(\mathbf{T}_n) \mathbf{\tilde{H}}^H(m) \mathbf{F} ^H  \mathcal{D}(\mathbf{F}_n) \right\}
\end{equation}
where the estimated channel matrix for the $m$-th OFDM block is denoted by  $\mathbf{\tilde{H}}(m)$, and
\begin{equation}\label{eq:Lambdam}
\mathbf{\Lambda}(m)=\mathcal{D}(\mathbf{F}\mathbf{\tilde{H}}(m)\mathbf{\tilde{H}}^H(m) \mathbf{F}^H) + \sigma_z^2 \mathbf{I}_N
\end{equation}
respectively, for $m=1, 2, \ldots, N$.

Note that, based on \eqref{eq:Rm} and \eqref{eq:Lambdam}, we can easily get update for the correlation matrices at each OFDM block $m$, and therefore the windowing filter based on the maximum average SINR criterion must be updated per block basis. This fact results into an increased computational cost, i.e. $\mathcal{O}(N^3)$. In order to overcome this, an iterative technique can be used for the update of the dominant eigenvector, which eventually, will converge to the desired eigenvector. In the following part of this subsection, we provide a description of the proposed adaptive algorithm.

Let us make the assumption that the correlation matrices $\mathbf{\bar{\Lambda}}(m)$ and $\mathbf{\bar{R}}(m)$ are stochastic, with $\mathcal{E}\{ \mathbf{\bar{\Lambda}}(m) \} = \mathcal{E}\{ \mathbf{\Lambda} \}$ and $\mathcal{E}\{ \mathbf{\bar{R}}(m) \} = \mathcal{E}\{ \mathbf{R} \}$ for all $m$. In this case, in order to update the estimation of the dominant eigenvector for the matrix $\mathbf{Q}(m) \triangleq \left( \mathbf{\bar{\Lambda}}^{-1/2}(m)\mathbf{\bar{R}}(m)  \mathbf{\bar{\Lambda}}^{-1/2}(m) \right)$, an iterative algorithm can be employed. Among the many available algorithms on this topic, we have chosen the algorithm in \cite{Oja85} since it provides strong performance guarantees. In particular, the authors in \cite{Oja85} have proposed a two-step iterative algorithm, which in our case can be expressed as follows,
\begin{eqnarray}\label{eq:oja_per(m)}
\mathbf{\bar{v}}(m) &\leftarrow& \mathbf{\bar{v}}(m-1) + \gamma_m \mathbf{Q}(m) \mathbf{\bar{v}}(m-1) \\
\mathbf{\bar{v}}(m) &\leftarrow& \frac{\mathbf{\bar{v}}(m)}{\Vert \mathbf{\bar{v}}(m) \Vert}
\end{eqnarray}
where $\mathbf{\bar{v}}(m)$ is the unknown dominant eigenvector and $\gamma_m$ is the step-size parameter. Once we have obtained the current block update of the dominant eigenvector $\mathbf{\bar{v}}(m)$, the optimal window can be computed by
\begin{equation}
\mathbf{\bar{w}}(m) = \mathbf{\bar{\Lambda}}^{-1/2}(m) \mathbf{\bar{v}}(m).
\end{equation}
When the channel is stationary, it is expected that after a sufficient number of OFDM blocks the vector $\mathbf{\bar{v}}(m)$ converges to the dominant eigenvector of $\mathcal{E} \{ \mathbf{\Lambda} \}^{-1/2} \mathcal{E} \{ \mathbf{R} \} \mathcal{E} \{ \mathbf{\Lambda} \}^{-1/2}$. For non-stationary channels, the parameter $\lambda$ can be set accordingly in order to track the potential variations of the channel statistics.

Note that, the step-size parameter $\gamma_m$ determines the convergence behavior of the adaptive algorithm. According to \cite{Oja85}, the convergence of the algorithm is guaranteed given that the following properties are satisfied, 
\begin{equation}\label{eq:gamma_conditions}
\gamma_m \ge 0, \sum_m \gamma_m^2 < \infty, \sum_m \gamma_m = \infty.
\end{equation}

The proposed adaptive windowing technique is summarized in Algorithm 1. The lines 1-9 are for the update of the sample correlation matrices, while the lines 10-14 are for the update of the dominant eigenvector, i.e. the windowing filter. 

\textit{Remark 1:} The expressions in eqs. \eqref{eq:Rm} and \eqref{eq:Lambdam} require the estimated CSI, as it is the case with the classical method based maximum SINR criterion. However in our case, the matrix $\mathbf{\tilde{H}}$ in eqs. \eqref{eq:Rm} and \eqref{eq:Lambdam}, is the channel correlation matrix in the frequency domain, and potentially, it can be estimated in a \textit{blind manner}, e.g. \cite{Moustakides06}, thus avoiding the costly operation of channel estimation. 

\textit{Remark 2:} For $K \ll N$, the complexity order of the proposed algorithm (Algorithm 1) is $\mathcal{O}(N^2)$. Although, the complexity cost of the proposed algorithm remains quadratic, it is an order of magnitude lower than that of the straightforward technique; thus the proposed technique could be beneficial for mobile applications with energy constraints  (i.e. Vehicle-to-Vehicle \cite{802.11p}).

\begin{algorithm}
\caption{Adaptive Windowing Technique}
\begin{algorithmic}[1]
\FOR {$m=1, 2, \ldots$}
\STATE \COMMENT{Update the sample correlation matrices}
\STATE $\mathbf{C}(m) \leftarrow  \mathbf{F} \mathbf{\tilde{H}}(m)$
\FOR  {$n=1, \ldots, N$}
\STATE $\mathbf{B}_n(m) \leftarrow \lambda \mathbf{B}(m-1) + \mathbf{C}(m) \mathcal{D}(\mathbf{T}_n) \mathbf{C}^H(m)$
\ENDFOR
\STATE $\mathbf{\bar{\Lambda}}(m) \leftarrow \lambda  \mathbf{\bar{\Lambda}}(m-1) + \mathcal{D}(\mathbf{C}(m)\mathbf{C}^H(m)) + \sigma_z^2 \mathbf{I}_N$
\STATE $\mathbf{\bar{R}}(m) \leftarrow \sum_{n=1}^N \mathcal{D}(\mathbf{F}_n) \mathbf{B}_n(m)\mathcal{D}(\mathbf{F}_n)$ 
\STATE $\mathbf{Q}(m) \leftarrow \left(\mathbf{\bar{\Lambda}}^{-1/2}(m)\mathbf{\bar{R}}(m)  \mathbf{\bar{\Lambda}}^{-1/2}(m)\right)$
\STATE \COMMENT{Update the dominant eigenvector}
\STATE $\mathbf{\bar{v}}(m) \leftarrow \mathbf{\bar{v}}(m-1) + \gamma_m \mathbf{Q}(m) \mathbf{\bar{v}}(m-1)$
\STATE $\mathbf{\bar{v}}(m) \leftarrow \frac{\mathbf{\bar{v}}(m)}{\Vert \mathbf{\bar{v}}(m) \Vert}$
\STATE $\mathbf{\bar{w}}(m) = \mathbf{\bar{\Lambda}}^{-1/2}(m) \mathbf{\bar{v}}(m)$
\ENDFOR
\end{algorithmic}
\end{algorithm}

\section{Performance Evaluation}\label{sec:Evaluation}

To evaluate the performance of the proposed technique, we consider an uncoded OFDM system with $N=16$ and QPSK constellation. The channel is modeled using tapped-delay-line model with three paths ($L=3$) and an exponential power delay profile. The path gain for each channel tap is independently generated from the Jakes' model \cite{Jakes}. The symbol estimation is obtained by the banded MMSE equalizer which is given by \eqref{eq:MMSE_equalizer}, with band size $K=3$, and regularization parameter $\epsilon=0.1$.

First, we evaluate the performance of the proposed method in terms of symbol-error-rate (SER) versus the signal-to-noise-ratio (SNR). We compare the performance of the proposed adaptive windowing technique with the following techniques: the Max-SINR and the Max-average-SINR based techniques \cite{Schniter04}, and the ground-truth version of the proposed technique, where the eigenvector is computed for each OFDM block via singular value decomposition. The parameter $\gamma_m$ of the proposed adaptive algorithm has been set equal to $\gamma_m = \frac{N^4}{m}$, which satisfies the conditions in \eqref{eq:gamma_conditions}.

Fig.~\ref{fig:spl_testcase_ber_snr} shows the SER performance comparison for the aforementioned techniques averaged over 15000 OFDM blocks, when the maximum normalized Doppler spread is $f_D=0.01$. Two cases for the channel state information are shown in Fig.~\ref{fig:spl_testcase_ber_snr}, i.e., perfect CSI (left figure) and imperfect CSI (right figure). For the imperfect case, the channel estimation errors were modeled according to $\mathbf{\tilde{H}}(m) =  \mathbf{H}(m) + \mathbf{O}(m)$, where $\mathbf{O}(m)$ is a $N \times N$ matrix with $[O]_{i,j}(m) \sim \mathcal{N}(0, \sigma_o^2)$. Moreover, we have considered erroneous estimation for the maximum Doppler frequency, i.e. $\tilde{f}_{d} =f_d + e$, where $e \sim \mathcal{U}(0, 0.01)$. 

We can observe that, for the perfect CSI case, all the three windowing techniques exhibit the same SER performance. On the contrary, for the imperfect case, only the proposed adaptive algorithm remains unaffected, while the other techniques exhibit a high error floor. Note that, the performance of the proposed technique is obtained after the algorithm reached the steady-state.

\begin{figure}[!t]
\centering
\includegraphics[width=3.8in]{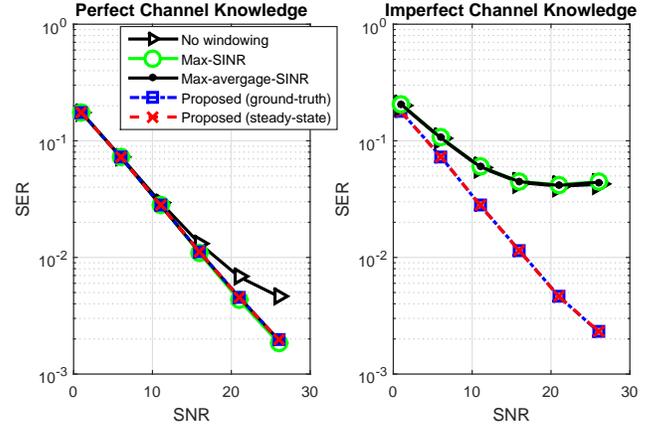}
\caption{Performance evaluation of the proposed technique in terms of SER.}
\label{fig:spl_testcase_ber_snr}
\vspace{-0.1in}
\end{figure}

In Fig.~\ref{fig:convergence}, we show the learning curve of the mean-square-error (MSE) with respect to the number of the OFDM blocks, averaged over 1000 Monte-Carlo realizations. In this case, in order to verify the convergence of the proposed adaptive technique, we have set a stationary channel, while the forgetting factor was set to $\lambda = 0.999$.
We can observe that the convergence speed of the proposed adaptive algorithm remains steady with respect to the SNR regime.

\begin{figure}[!t]
	\centering
	\includegraphics[width=3.8in]{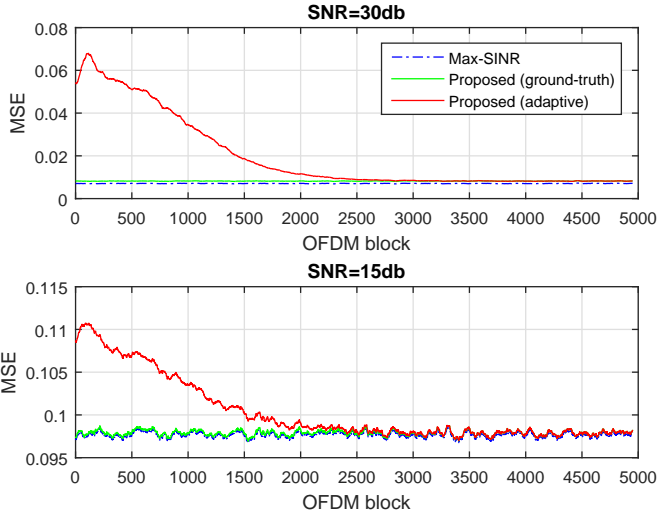}
	\caption{MSE learning curve of the proposed adaptive technique for two SNR regimes; $15dB$ for the medium regime and $30dB$ for the high regime.}
	\label{fig:convergence}
	\vspace{-0.1in}
\end{figure}

Fig.~\ref{fig:tracking} shows the tracking capabilities of the proposed adaptive technique. In our scenario, we assume that after 6000 OFDM blocks a sudden change occurs at the channel parameters, i.e. the maximum Doppler spread increases from 0.001 to 0.005. It can be seen that after 2000 OFDM blocks, the proposed adaptive algorithm converges to the optimal steady-state bound. Note that, the forgetting factor was been set to $\lambda = 0.98$, while $SNR=30dB$. 

\begin{figure}[!t]
\centering
\includegraphics[width=3.8in]{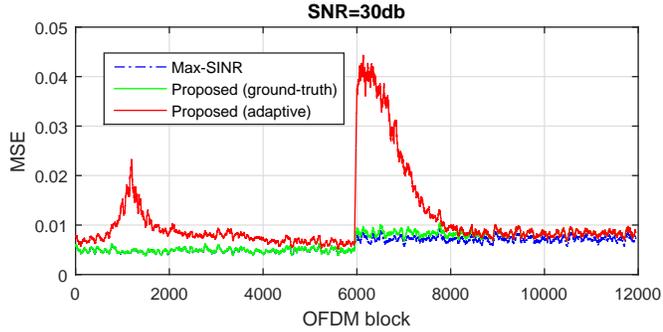}
\caption{MSE tracking curve.}
\label{fig:tracking}
\vspace{-0.1in}
\end{figure}

%Specifically, the key point here is that the involved channel matrix $\mathbf{H}$ refers specifically to the convolution channel matrix. This property may give us the option of \textit{blind estimation} of the matrix $\mathbf{R}$, and thus avoiding the estimation of the channel coefficients.

%If the  channel statistics do not vary over time, then the algorithm will converge to the maximum eigenvector of the matrix $\mathbf{Q}$. On the other hand, if the channel statistics vary over time, then the proposed algorithm is able to track the varying maximum eigenvector.

\section{Conclusion}
In this letter, a novel adaptive technique for enhancing the banded structure of the OFDM channel matrix has been proposed. We have considered the cases of unknown and erroneous channel statistics, where the conventional methods fail to operate effectively. While the straightforward optimal solution to this problem requires high complexity, the proposed adaptive algorithm is able to track the optimal window requiring an order of magnitude lower complexity.

\section*{Appendix}

\subsection{Derivation of Proposition~\ref{proposition:1}}

Let us first consider the numerator of \eqref{eq:maxSINR}, which can be expressed as follows \cite[page 110, eq. (3.6.1)]{Johnson89},
$
\Vert \mathbf{T} \circ \left( \mathcal{C}(\mathbf{w}) \mathbf{H} \right) \Vert_{F}^2 = \left\Vert \sum_{n=1}^N \mathbf{E}_{nn} \mathcal{C}(\mathbf{w}) \mathbf{H} \mathcal{D}(\mathbf{T}_n) \right\Vert_{F}^2
$, where $\mathbf{E}_{nn}$ is a matrix with one at the $n$-th row and $n$-th column and zeros elsewhere. This special matrix has the following property
$\mathbf{E}_{nn}\mathbf{E}_{mm} = \left\{\begin{array}{c}
\mathbf{E}_{nn}, n=m \\
\mathbf{0}_N, n \neq m
\end{array}\right.$
and thus it can be shown that
\begin{equation}
\left\Vert \sum_{n=1}^N \mathbf{E}_{nn} \mathcal{C}(\mathbf{w}) \mathbf{H} \mathcal{D}(\mathbf{T}_n) \right\Vert_{F}^2 = \sum_{n=1}^N \Vert \mathbf{E}_{nn} \mathcal{C}(\mathbf{w}) \mathbf{H} \mathcal{D}(\mathbf{T}_n) \Vert_{F}^2.
\end{equation}
Then, based on properties of the Frobenius norm
\cite{Golub96}, we have that $\Vert \mathbf{T} \circ \left( \mathcal{C}(\mathbf{w}) \mathbf{H} \right) \Vert_{F}^2 = \sum_{n=1}^N tr\left(\mathbf{E}_{nn} \mathcal{C}(\mathbf{w}) \mathbf{H} \mathcal{D}(\mathbf{T}_n)\mathcal{D}(\mathbf{T}_n))^H \mathbf{H}^H \mathcal{C}(\mathbf{w})^H \mathbf{E}_{nn}^H \right)
$. Since $\mathcal{C}(\mathbf{w}) = \mathbf{F}\mathcal{D}(\mathbf{w}) \mathbf{F}^H$, it is straightforward to show that $\mathbf{E}_{nn} \mathcal{C}(\mathbf{w})=\mathbf{E}_{nn}  \mathbf{F}\mathcal{D}(\mathbf{w}) \mathbf{F}^H= \mathbf{J}_n(\mathbf{w}) \mathcal{D}(\mathbf{F}_n)\mathbf{F}^H$, where $\mathbf{J}_n(\mathbf{w})$ is a matrix with zero rows except for the $n$-th row which is equal to $\mathbf{w}$, and $\mathcal{D}(\mathbf{F}_n)$ is the diagonal matrix whose diagonal equals the $n$-th row of the DFT matrix. Moreover, it is true that
$tr(\mathbf{J}_n(\mathbf{w}) \mathbf{X} \mathbf{J}_n(\mathbf{w})^H) = \mathbf{w}^H \mathbf{X} \mathbf{w}$. Therefore, we end up with the following expression
\begin{equation}
\Vert \mathbf{T} \circ \left( \mathcal{C}(\mathbf{w}) \mathbf{H} \right) \Vert_{F}^2 
= \mathbf{w}^H \mathbf{R} \mathbf{w} \label{eq:wRw}
\end{equation}
where $\mathbf{R} = \sum_{n=1}^N \mathcal{D}(\mathbf{F}_n) \mathbf{C} \mathbf{F}^H \mathbf{\Lambda}(\mathbf{T}_n) \mathbf{F}  \mathbf{C}^H \mathcal{D}(\mathbf{F}_n)$.
Considering now the denumerator of \eqref{eq:maxSINR}, recall that it is decomposed into the interference and the AWGN terms. The interference energy term can be written as
\begin{equation}\label{eq:denomof6}
\Vert \mathbf{T}^c\circ \left( \mathcal{C}(\mathbf{w}) \mathbf{H}\right) \Vert_{F}^2 = \Vert \mathcal{C}(\mathbf{w}) \mathbf{H} \Vert_{F}^2 - \Vert \mathbf{T} \circ \left( \mathcal{C}(\mathbf{w}) \mathbf{H} \right) \Vert_{F}^2
\end{equation}
where
\begin{equation}
\Vert \mathcal{C}(\mathbf{w}) \mathbf{H} \Vert_{F}^2 
= \Vert \mathbf{F}\mathcal{D}(\mathbf{w}) \mathbf{F}^H \mathbf{F} \mathbf{C} \mathbf{F}^H \Vert_{F}^2 
=  \mathbf{w}^H \mathcal{D}(\mathbf{C} \mathbf{C}^H) \mathbf{w} \label{eq:wDCCw}
\end{equation}
By combining \eqref{eq:wRw} and \eqref{eq:wDCCw} into \eqref{eq:denomof6} we have that
\begin{equation}
\Vert \mathbf{T}^c\circ \left( \mathcal{C}(\mathbf{w}) \mathbf{H}\right) \Vert_{F}^2 = \mathbf{w}^H \left( \mathcal{D}(\mathbf{C}\mathbf{C}^H) - \mathbf{R} \right) \mathbf{w}
\end{equation}
Finally, by using that $\Vert \mathcal{C}(\mathbf{w}) \Vert_{F}^2 = \mathbf{w}^H \mathbf{w}$, we have that $\Vert \mathbf{T}^c\circ \left( \mathcal{C}(\mathbf{w}) \mathbf{H}\right) \Vert_{F}^2 + \sigma_z^2 \Vert \mathcal{C}(\mathbf{w}) \Vert_{F}^2 = \mathbf{w}^H \left(\sigma_z^2 \mathbf{I}_N + \mathcal{D}(\mathbf{C}\mathbf{C}^H) - \mathbf{R} \right) \mathbf{w}$.

\subsection{Proof of Proposition~\ref{proposition:2}}

Eq.~\eqref{eq:maxSINR2} can be written as
\begin{eqnarray}\label{eq:pencil}
\mathbf{R}\mathbf{w}^{\star} &=& \eta_m (\mathbf{\Lambda} - \mathbf{R}) \mathbf{w}^{\star} \\
\Rightarrow \mathbf{\Lambda}^{-1/2} \mathbf{R} \mathbf{\Lambda}^{-1/2} \mathbf{v}^{\star} &=& \kappa_m \mathbf{v}^{\star}
\end{eqnarray}
where $\kappa_m = \frac{\eta_m}{1+\eta_m}$, where $\eta_m$ is the maximum eigenvalue. Note that the matrices $\mathbf{R}$ and $\mathbf{\Lambda}-\mathbf{R}$ are positive semi-definite, since they can be written as Gram matrices, i.e.
\begin{eqnarray}
\mathbf{R} = \sum_{n=1}^N \mathcal{D}(\mathbf{F}_n) \mathbf{F} \mathbf{H} \mathcal{D}(\mathbf{T}_n) \mathbf{H}^H \mathbf{F}^H \mathcal{D}(\mathbf{F}_n)^H = \mathbf{U}^H \mathbf{U}
\end{eqnarray}
and $  \mathbf{\Lambda} - \mathbf{R} = \left[ \begin{matrix} \left( \mathbf{\Lambda}^{1/2} \right)^H & j \mathbf{U}^H \end{matrix} \right] \left[ \begin{matrix} \mathbf{\Lambda}^{1/2} \\ j \mathbf{U} \end{matrix} \right]$. Therefore, $\lambda_{max} \ge 0$ and the function $f(\eta)=\frac{\eta}{1+\eta}$ is strictly increasing, and thus the eigenvector of the $\kappa_m$-th eigenvalue corresponds to the eigenvector of $\eta_m$. Taking the statistical mean of the involved matrices in \eqref{eq:pencil} we have
\begin{align*}
& \mathcal{E} \{ \mathbf{R} \} \mathbf{w}^{\star} = \eta_m' \mathcal{E} \{ (\mathbf{\Lambda} - \mathbf{R}) \} \mathbf{w}^{\star} \\
\Rightarrow  & \mathcal{E}\{ \mathbf{\Lambda}\}^{-1/2} \mathcal{E}\{\mathbf{R}\} \mathcal{E}\{\mathbf{\Lambda}\}^{-1/2} \mathcal{E}\{\mathbf{\Lambda}\}^{1/2} \mathbf{w}^{\star} = \kappa_m' \mathcal{E}\{\mathbf{\Lambda}\}^{1/2} \mathbf{w}^{\star} \\
\Rightarrow & \mathcal{E} \{\mathbf{\Lambda}\}^{-1/2} \mathcal{E} \{\mathbf{R}\} \mathcal{E} \{\mathbf{\Lambda}\}^{-1/2} \mathbf{v}^{\star} = \kappa_m' \mathbf{v}^{\star}
\end{align*}
which results in \eqref{eq:maxSINR4}.

% or


%\section*{Acknowledgment}

\ifCLASSOPTIONcaptionsoff
  \newpage
\fi



%\bibliographystyle{IEEEtran}

\newpage
\bibliographystyle{IEEEtran}
\bibliography{thebibliography}

\end{document}